\documentclass[10pt, conference, letterpaper]{IEEEtran}
\makeatletter
\def\ps@headings{%
\def\@oddhead{\mbox{}\scriptsize\rightmark \hfil \thepage}%
\def\@evenhead{\scriptsize\thepage \hfil \leftmark\mbox{}}%
\def\@oddfoot{}%
\def\@evenfoot{}}
\makeatother
\usepackage{amsmath}
\usepackage{amsthm}
\usepackage{graphicx}
\usepackage{subfigure}
\usepackage{mathrsfs}
\usepackage{algorithm}
\usepackage{color}
\usepackage{url}
\usepackage{enumerate}
\usepackage{booktabs}
\usepackage{amsmath}
\usepackage{amssymb}
\usepackage[english, greek]{babel}
\usepackage{algpseudocode}
\usepackage{flushend}

\newtheorem{theorem}{Theorem}

\newtheorem{definition}{Definition}

\newtheorem{property}{Property}

\DeclareMathOperator*{\argmin}{arg\,min}

\begin{document} 
%
\IEEEoverridecommandlockouts
\selectlanguage{english}

\title{Multi-resource Energy-efficient Routing in Cloud Data Centers with Networks-as-a-Service\thanks{This work was supported in part by Ministerio de Economia y Competitividad grant TEC2014- 55713-R, Regional Government of Madrid (CM) grant Cloud4BigData (S2013/ICE-2894, co-funded by FSE \& FEDER), NSF of China grant 61520106005, and European Commission H2020 grants ReCred and NOTRE.}}



%

\author{
\IEEEauthorblockN{
Lin Wang\IEEEauthorrefmark{1}\IEEEauthorrefmark{4},
Antonio Fern\'andez Anta\IEEEauthorrefmark{3},
Fa Zhang\IEEEauthorrefmark{2},
Jie Wu\IEEEauthorrefmark{5}
Zhiyong Liu\IEEEauthorrefmark{1}\IEEEauthorrefmark{6}
}
\IEEEauthorblockA{
\IEEEauthorrefmark{1}Beijing Key Laboratory of Mobile  Computing and Pervasive Device, ICT, CAS, China}  
\IEEEauthorblockA{
\IEEEauthorrefmark{2}Key Lab of Intelligent Information Processing, ICT, CAS, China} 
\IEEEauthorblockA{
\IEEEauthorrefmark{3}IMDEA Networks Institute, Spain}
\IEEEauthorblockA{
\IEEEauthorrefmark{4}University of Chinese Academy of Sciences, China}
\IEEEauthorblockA{
\IEEEauthorrefmark{5}Department of Computer and Information Sciences, Temple University, USA}
\IEEEauthorblockA{
\IEEEauthorrefmark{6}State Key Laboratory for Computer Architecture, ICT, CAS, China}
}


\maketitle

\begin{abstract}
With the rapid development of software defined networking and network function virtualization, researchers have proposed a new cloud networking model called Network-as-a-Service (NaaS) which enables both in-network packet processing and application-specific network control. In this paper, we revisit the problem of achieving network energy efficiency in data centers and identify some new optimization challenges under the NaaS model. Particularly, we extend the energy-efficient routing optimization from single-resource to multi-resource settings. We characterize the problem through a detailed model and provide a formal problem definition. Due to the high complexity of direct solutions, we propose a greedy routing scheme to approximate the optimum, where flows are selected progressively to exhaust residual capacities of active nodes, and routing paths are assigned based on the distributions of both node residual capacities and flow demands. By leveraging the structural regularity of data center networks, we also provide a fast topology-aware heuristic method based on hierarchically solving a series of vector bin packing instances. Our simulations show that the proposed routing scheme can achieve significant gain on energy savings and the topology-aware heuristic can produce comparably good results while reducing the computation time to a large extent.

\end{abstract}

\section{Introduction}
\label{sec:intro}

With the widespread adoption of cloud computing, enormous large-scale data centers have been deployed for companies like Google, Microsoft, and Amazon, to provide online services including searching and social networking. Generally speaking, data centers are consolidated facilities holding tens of thousands servers that are connected by a well-structured network termed data center network (DCN). Despite some designs that rely on specialized hardware and communication protocols, most of the DCN architectures leverage commodity Ethernet switches and routers to interconnect servers, and thus are compatible with TCP/IP applications.

As inter-node communication bandwidth is the principal bottleneck in data centers, there has been a large body of work on optimizing the performance of DCNs (e.g., \cite{Al-Fares_Radhakrishnan-2010, Benson_Anand-2011}). However, in order to apply these proposals to production DCNs, a lot of effort has to be undertaken, including both hardware- and software-end modifications. This is due to the specific deployment settings designed for the routing and forwarding protocols used in current data centers. As a result, incremental implementations are usually not achievable and significant effort has to be made for every single design.

The situation has been changed with the evolution of network architecture. On the one hand, researchers proposed to decouple control plane from data plane, which enables rapid innovation in network control. This idea then naturally led to the advent of Software Defined Networking (SDN). Instead of having all network nodes to run the routing protocol in a distributed manner, SDN abstracts the network control functionality to a logically centralized controller. The routing decisions are then made by the controller and will be pushed to network nodes through a well-defined Application Programming Interface (API) such as OpenFlow \cite{of}. As a result, the optimization work can be totally implemented in the controller, which needs only very basic software modification in the event of network changes. On the other hand, the innovation in data plane has also been sped up by the technology called Network Function Virtualization (NFV) where packets are handled by software-based entities on general-purpose servers with network functions virtualized.

\begin{figure}[!t]
\centering
\includegraphics[scale=0.31]{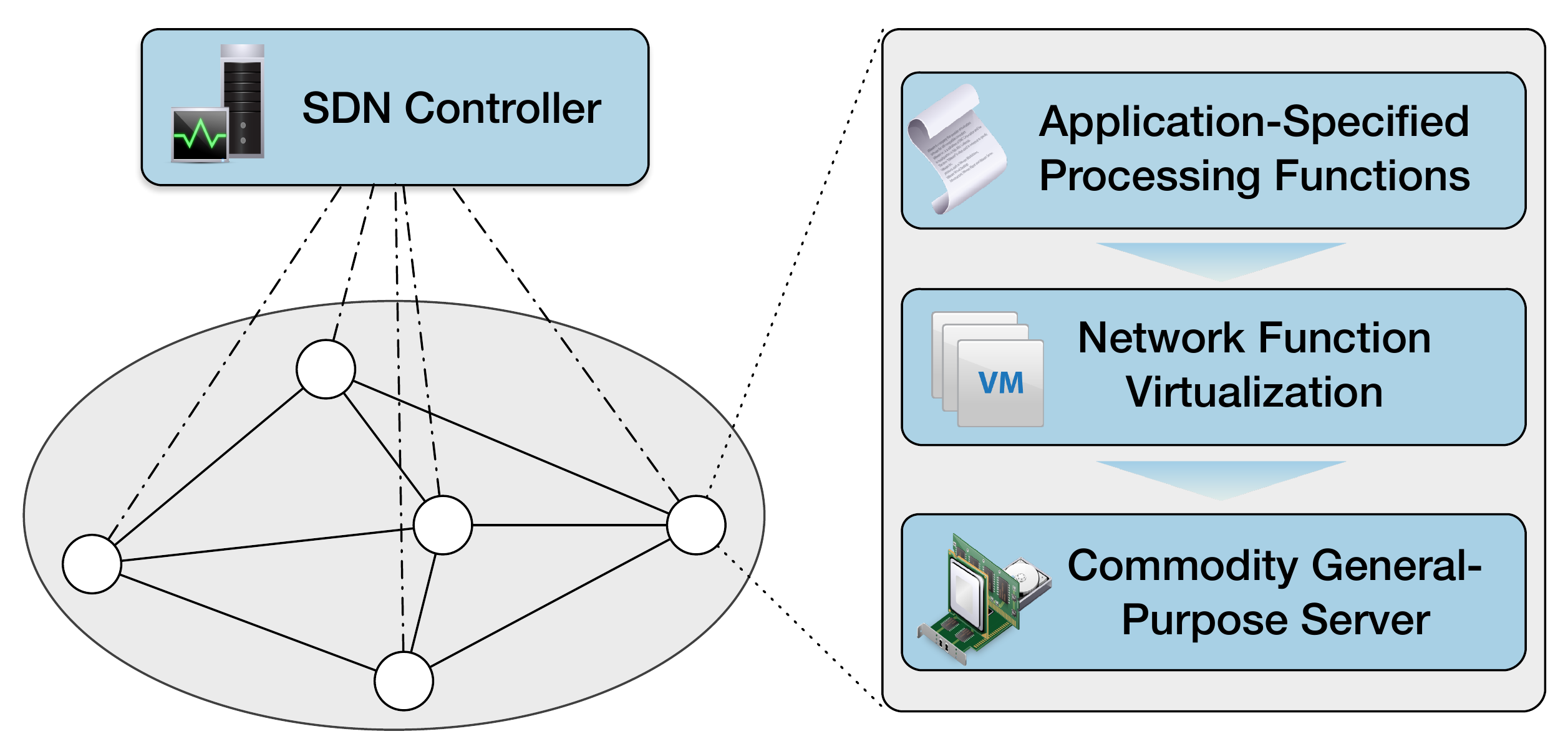}
\caption{\label{fig:naas}An overview of a general NaaS implementation in a typical cloud data center. The control plane is separated from the data plane by the SDN while the data plane can be customized through NFV.}
\vspace{-0.4cm}
\end{figure}

Taking advantage of the advancement of both control plane and data plane in networks, a new cloud networking model called \emph{Network-as-a-Service} (NaaS) has been recently proposed \cite{Benson_Akella-2011, Costa-Migliavacca}. An overview of a general NaaS implementation can be found in Fig.~\ref{fig:naas}.  In the NaaS model, packet forwarding decisions are implemented based on specific application needs. Moreover, NaaS enables the design of in-network packet modification and thus in-network services, such as data aggregation, stream processing or caching can be specified by upper-layer applications. Based on this new networking model, several working examples have been studied, including in-network aggregation for big data applications \cite{Costa_Donnelly-2012}.

Although there is still a range of research challenges for a widespread adoption of NaaS, we observe that some fundamental problems have emerged from this new model. It is recognized that no matter what networking model is employed, problems such as load balancing and energy saving always possess their importance. However, compared to the traditional packet-forwarding-centric model, NaaS brings new challenges to these problems by allowing in-network packet operations, making existing solutions not efficient or even not applicable to these problems any more.

Particularly, we revisit the problem of achieving network energy efficiency in data centers and identify some new challenges under the NaaS model. With packets being processed by general-purpose servers, energy-related issues become more prominent. The energy saving problem in DCNs has been widely studied and most proposals are based on the general approach of consolidating network flows and turning off unused network elements (e.g., \cite{Heller-Seetharaman, Vasic-Bhurat, Wang-Yao, Wang-Zhang-jsac}). In packet forwarding networks, link utilization is the most important criterion for flow consolidation. However, this is no longer valid under the NaaS model, where a network node can be congested not only by data communications, but also by the overloading of other resources such as processing units or memory. Without considering other resources, a link utilization oriented consolidation of flows may lead to very bad resource saturation at some network nodes and to serious network instability. Therefore, it is essential to take into account multiple resources when making routing decisions under the NaaS model. To the best of our knowledge, this is the first research attempt towards multi-resource traffic engineering.


The main contributions of this paper are as follows: $\mathit{i})$ we identify new research challenges for conventional optimization problems under the NaaS model, and characterize the network energy saving problem through a detailed model with complexity analyzed; $\mathit{ii})$ we propose a greedy routing scheme where path selection is done based on the distributions of node residual capacities and flow demands; $\mathit{iii})$ by utilizing the structural property of DCNs, we provide a topology-aware heuristic which can accelerate problem solving while producing comparably good results; $\mathit{iv})$ we validate the efficiency of the proposed algorithms by extensive simulations and show that significant energy efficiency gain in NaaS-enabled DCNs can be achieved by the techniques proposed in this paper.

The remainder of this paper is organized as follows. Section~\ref{sec:related} summarizes the related work. Section~\ref{sec:problem} gives the model and the definition of the problem, as well as some complexity analysis. Section~\ref{sec:alg} proposes a greedy routing scheme and Section~\ref{sec:heuristic} presents a topology-aware heuristic. Section~\ref{sec:simu} examines the performance of the proposed algorithms by simulations. Section~\ref{sec:conc} concludes the paper.


\section{Related Work}
\label{sec:related}

In this section, we revise from several viewpoints the existing work related to our study.

{\bf Software defined networking.} 
The high-level coupling of the control plane and the data plane in traditional networks brings very high complexity to network management and leads to a very slow pace of development and evolution of network functionalities due to the reliance on proprietary hardware. SDN is forced to solve these problems by changing the design and management of a network in the following two ways: in an SDN, the control plane and the data plane are clearly separated; the control plane is logically consolidated. The control plane thus exercises a full view of the network and can be implemented with a single software control program. Through a well-defined API, the control plane carries out direct control to push decisions over multiple data-plane elements in the network. The logically centralized control can facilitate most network applications including network virtualization \cite{Sherwood-Gibb}, server load balancing \cite{Wang-Butnariu}, and energy-efficient networking \cite{Heller-Seetharaman, Wang-Zhang-jsac}, which would require enormous efforts to be implemented in a totally distributed environment.

{\bf Software packet processing.} Recently, researchers have argued for building evolvable networks, whose functionality changes with the needs of its users and is not tied to particular hardware vendors \cite{Fall_Iannaccone-2011}. This is called general-purpose networking, where a network-programming framework is run on top of commodity general-purpose hardware.
On the one hand, there have been several research prototypes demonstrating that general-purpose hardware is capable of high-performance packet processing when packets are subjected to single particular type of processing, such as IP forwarding \cite{Fall_Iannaccone-2011} or cryptographic operations \cite{Jang_Han-2011}. It has also been shown in \cite{Dobrescu_Argyraki-2012} that such a software packet processing platform can achieve predictable performance while running various packet-processing applications and serving multiple clients with different needs. On the other hand, Niccolini \emph{et al.} \cite{Niccolini-Iannaccone-2012} developed software mechanisms that exploit the underlying hardware's power management features for more energy-efficient packet processing in software routers.

{\bf Network-as-a-Service.} With the development of SDN and software packet processing, a new networking model called NaaS has been recently proposed \cite{Benson_Akella-2011, Costa-Migliavacca}. Under this model, the network is conducted based on an SDN implementation where a centralized controller takes charge of flow management, including routing paths assignment and network functions interposition. The network is comprised of general-purpose servers with multiple network ports connected by high-speed links. In each node in the network, network functions for packet processing are virtualized and can be invoked through the controller by upper-layer applications according to their needs. 
Several research attempts have already been made to adopt NaaS in real date centers \cite{Costa_Donnelly-2012}.

{\bf Energy-efficient DCN.} The topic of achieving energy-efficient data center networks has been extensively studied. Research efforts are concentrated on the following two categories. One is designing new architectures with less network devices involved while providing similar end-to-end connectivity \cite{Abts-Marty}. The other is applying traffic engineering techniques to consolidate network flows and switch off unused network elements (switches or links). The seminal work in this category is the concept of ElasticTree proposed by Heller \emph{et al.} \cite{Heller-Seetharaman}, which is a network-wide power manager that dynamically adjusts the set of active network devices to satisfy changing traffic loads in DCNs. Shang \emph{et al.} \cite{Shang-Li} discussed how energy can be saved by energy-aware routing with negligible performance degradation in high-density DCNs. Some follow-up works include REsPoNse \cite{Vasic-Bhurat}, CARPO \cite{Wang-Yao} and GreenDCN \cite{Wang-Zhang-jsac}. 


{\bf Multi-resource allocation.} Resource allocation in computing systems has been widely discussed under the limit of a single resource, such as CPU time and link bandwidth. While multi-resource allocation is considered in cloud computing systems, it is usually carried out with a slot-based single resource abstraction. Recently, researches have made some efforts (e.g., \cite{Joe-Wong_Sen-2012, Wang_Li-2014}) towards multi-resource fair sharing under the Dominant Resource Fairness (DRF) model proposed by Ghodsi \emph{et al.} \cite{Ghodsi_Zaharia-2012}. On the network side, DRF-based approaches (e.g., \cite{Ghodsi_Sekar-2012, Wang_Liang-2014}) have also been proposed to achieve multi-resource fair queueing in software packet processors.

Compared to these studies, our work possesses its uniqueness in the sense that we are the first to target the problem of achieving network-wide energy efficiency under multi-resource settings based on the NaaS model.

\section{Problem Statement}
\label{sec:problem}


In the server end, the problem of finding the minimum number of servers to accommodate a given set of tasks whose resource requirements are characterized by vectors is defined as the Vector Bin Packing (VBP) problem. The best known solution for the general form of the problem is a $(\ln K) + 1+\epsilon$ approximation for any $\epsilon > 0$, provided by Basal \emph{et al.}, \cite{Bansal-Caprara-2006} where $K$ is the number of dimensions of each item. While the single-resource network energy optimization problem has been well-studied, very little attention has been received by the energy-efficient routing problem in networks with multiple resources. With an emerging trend of software packet processing in networks, this problem has raised its significance. In the following, we provide a formal modeling of the problem and examine its complexity.

\subsection{Preliminary Notations}

We abstract a given software packet-processing network as graph $\mathcal{G} = \{\mathcal{V}, \mathcal{E}\}$, where $\mathcal{V}$ is the set of $N$ nodes, each of which represents a general-purpose server with software packet processing functionalities, and $\mathcal{E}$ is the set of undirected edges representing the network links. Each node $v \in \mathcal{V}$ has limited amounts of $K$ different types of hardware resources, namely CPU, memory, and network bandwidth, to name a few. The total amount of type-$k$ resource is constrained by a positive capacity $C_{v,k}$ $(k \in \{1,2,...,K\})$. Due to the fact that packet-processing networks are usually constructed using commodity general-purpose servers, it is reasonable to assume that all the nodes in $\mathcal{V}$ are identical. Thus, for all $v \in \mathcal{V}$, we assume $C_{v,k} = C_k$ for all $k \in  \{1,2,...,K\}$.

We define a \emph{flow} as a sequence of data packets that possess the same entities in the packet headers such as the same source and destination IP addresses. Suppose we are given a set of $M$ flow demands $\mathcal{D} = \{\mathsf{d}_1, \mathsf{d}_2,...,\mathsf{d}_M\}$. The packets from the same flow $\mathsf{d}_m$ will be routed following a single path in order to avoid packet reordering at the destination. For all the packets from a given flow, a processing procedure is defined on every node on the flow's routing path, which is used to carry out some per-flow computation to the payloads of packets, e.g., intercepting packets on-path to implement opportunistic caching strategies \cite{Costa-Migliavacca}. Due to the fact that the data carried by the packets from the same flow generally possess the same structure (e.g., same packet size), we assume that (nearly) the same amount of computation will be applied to the packets from the same flow. As a result, we have to keep (almost) the same reservation across each type of resource on every node on the path for each flow. Each flow $\mathsf{d}_m$ is represented by a three-tuple $(v_m^{\mathrm s}, v_m^{\mathrm t}, \vec{R}_m)$ where $v_m^{\mathrm s}$ and $v_m^{\mathrm t}$ are the source and destination respectively, while $\vec{R}_m$ is a $K$-dimensional vector $(r_{m,1},...,r_{m,K})$ describing the amounts of resources in all types required (and reserved) for a node to process the packets from flow $\mathsf{d}_m$. These resource demands can be obtained by applying the same technique used in \cite{Ghodsi_Sekar-2012}. For the sake of simplicity and without loss of generality, we assume that the $r_{m,k}$ for $m \in  \{1,2,...,M\}$ are \textit{normalized} by $C_k$ for any $k \in  \{1,2,...,K\}$, i.e. $\vec{R}_m \in [0,1]^K$.

To quantify the performance of approximations, we term $\gamma$ as the \emph{performance ratio} of an algorithm for a minimization problem if the objective values in the solutions provided by the algorithm are upper-bounded by $\gamma$ times the optimal.

\subsection{Problem Formulation}

Using the introduced notations, the energy-efficient multi-resource routing problem can be formally defined as follows. For a vector $\vec{x}$, we denote by $||\vec{x}||_\infty$ the standard $\ell_\infty$ norm.

\begin{definition}[{\bf ENERGY-EFFICIENT MULTI-RESOURCE ROUTING (EEMR)}]
Given a network $\mathcal{G} = (\mathcal{V}, \mathcal{E})$ and a set of $M$ flows $\mathsf{d}_1,...,\mathsf{d}_M$ whose demands are characterized by $\vec{R}_1,...,\vec{R}_M$ from $[0,1]^K$, find a path $\mathcal{P}_m$ from $v_m^{\mathrm{s}}$ to $v_m^{\mathrm{t}}$ for each flow $\mathsf{d}_m$ such that $||\vec{A}_v||_\infty \leq 1$ for $v \in \mathcal{V}$ where $\vec{A}_v = \sum_{m:v \in \mathcal{P}_m} \vec{R}_m$ is the aggregation of the resource requirement vectors of flows that are routed through node $v$. The objective is to minimize $|\mathcal{Q}|$ where $\mathcal{Q} = \{v ~|~v \in \mathcal{V} \wedge \vec{A}_i \neq (0,...,0)\}$ is the set of nodes that are used to carry flows.
\end{definition}

The EEMR problem can be formulated as a Mixed Integer Program (MIP) in the following way. We introduce two binary variables $x_{m,v}$ and $y_v$. The binary variable $x_{m,v}$ indicates whether flow $\mathsf{d}_m$ is routed through node $v$ and $y_v$ indicates whether node $v$ is active or not. Our objective is to minimize the number of active nodes.\footnote{As the static power consumption of a node is dominant, we only consider using power-down based strategy as the main energy saving mechanism.} Note that we have an implicit assumption that feasible solutions are always achievable, that is, the network with the designed capability is able to handle the given traffic demands.
\begin{equation}
\begin{aligned}
(\mathbb{P}_1)~~~~&\text{minimize}~~  \sum_{v \in \mathcal{V}} y_v  \\
\text{subject to} \\
& ||\sum_{m \in \{1,2,...,M\}} \vec{R}_m \cdot x_{m,v}||_{\infty} \leq 1 ~~~~~ v \in \mathcal{V}\\
& x_{m,v} \leq y_v ~~~~~~~~~~~~~~~~~~ v \in \mathcal{V}, 1 \leq m \leq M \\
& x_{m,v}, y_{v} \in \{0,1\} ~~~~~~~~~~ v \in \mathcal{V}, 1 \leq m \leq M \\
& x_{m,v}:~\text{flow~conservation}
\end{aligned}
\nonumber
\end{equation}
The constraints of program $\mathbb{P}_1$ are as follows: the first constraint states that the flows routed through the same node do not exceed the node resource dimensions; the second constraint tells whether a node is active or not; the third constraint ensures that each flow can only follow a single path. Flow conservation on $x_{m,v}$ forces that the nodes that flow demand $\mathsf{d}_m$ is routed through form a path between $v_m^{\mathrm s}$ and $v_m^{\mathrm t}$ in the network.

Note that when $K = 1$, $\mathbb{P}_1$ corresponds to the general \emph{capacitated network design} problem which has been widely studied. For the uniform link capacitated version of the problem, Andrews, Antonakopoulos and Zhang \cite{Andrews-Anto-2010} provided a polylogarithmic approximation when the capacity on each link is allowed to be exceeded by a polylogarithmic factor. Recently, \cite{Krish-Naga-2014} explored the multicommodity node-capacitated network design problem and provided a $\mathcal{O}(\log^5 n)$-approximation with $\mathcal{O}(\log^{12} n)$ congestion. However, none of the studies can provide high-quality approximations with capacity constraints that are inviolable. This is mainly because with strict capacity constraints, finding out whether there is a feasible solution for the problem is already NP-hard.

\subsection{Complexity Analysis}

In contrast with the traditional energy-efficient routing (i.e., capacitated network design) problem, EEMR extends the concept of ``load" from single-dimensional to multidimensional, which makes the problem even computationally harder. In general, we have the following complexity results.
\begin{theorem}
Solving the EEMR problem is NP-hard.
\end{theorem}

\begin{proof}
The proof is conducted on a polynomial-time reduction from the VBP problem which is known to be NP-hard. Assume we are given an arbitrary instance of VBP and now we reduce it to the EEMR problem in the following way: each bin in the VBP instance is a node in the network for EEMR and each node is connected with two extra nodes $\mathsf{src}$ and $\mathsf{dst}$. Each item in the VBP instance represents a flow which originates from $\mathsf{src}$ and ends at $\mathsf{dst}$ and has resource requirements characterized by the vector for the item. Then, straightforwardly, if we obtain an optimal solution for EEMR, this solution will correspond to an optimal solution to VBP with the same structure. As a result, any polynomial-time algorithm that optimally solves EEMR can also be used to solve VBP optimally, which contradicts with the fact that VBP is NP-hard.
\end{proof}
\begin{theorem}
There is no asymptotic PTAS for the EEMR problem unless P=NP.
\end{theorem}

This is directly applied from the fact that VBP with $K \geq 2$ is know to be APX-hard which implies that there is no asymptotic PTAS for it \cite{Woeginger-1997}. From the above reduction we know that actually VBP is a special case for EEMR, meaning that EEMR has at least the same complexity as VBP.

\section{Energy-efficient Multi-resource Routing}
\label{sec:alg}

The complexity analysis results show that the EEMR problem is NP-hard, for which no existing exact solutions can scale to the size of current data center networks. Therefore, we resort to an intuitive approach that can provide suboptimal solutions very quickly. We detail our design in this section. 

\subsection{Key Observations}

We propose a greedy routing scheme to solve the energy-efficient multi-resource routing problem. The general idea is to use as few as possible nodes to carry all the traffic flows while maintaing the capacity constraints in all resource dimensions. More specifically, our design is based on the following two observations: $\mathit{i}$) flows preferably follow paths that consists of more \emph{active} nodes (that already carry some traffic) as this will introduce less extra energy consumption to the network; $\mathit{ii}$) it is important to allocate routes for flows on the active nodes such that all dimensions of the resources in every active node can be fully utilized. 

The second observation is a new concern steaming from the multi-resource context. In the single-resource case, the only criterion for the efficiency of a node is its resource utilization, i.e., the carried traffic divided by the total capacity. As a result, steering flows to those nodes with low utilizations will lead to an energy-efficient routing solution. However, this approach is not applicable to the multi-resource case. With multiple dimensions of resources, it is not clear how to define the resource utilization of a node, thus we will not be able to make routing decisions based on node utilizations. For instance, given two load vectors $\langle 0.6, 0.4, 0.1 \rangle$ and $\langle 0.4, 0.4, 0.3 \rangle$ of a node, and a flow demand vector $\langle 0.1, 0.3, 0.4 \rangle$, the resulted loads of this node when routing the given flow under the two different load levels are $\langle 0.7, 0.7, 0.5 \rangle$ and $\langle 0.5, 0.7, 0.7 \rangle$, which are not directly comparable. In order to step over this obstacle, we will provide a measuring method based on the distributions of node residual capacities and flow demands.

\subsection{The Routing Scheme}

The pseudocode of the routing scheme is shown in Algorithm~\ref{alg:mrg}. The algorithm runs progressively. In each iteration, it first tries to use only the set of active nodes. By searching the flow demand list, it tries to find out a candidate flow to route on the subnetwork $\mathcal{G}_a$ composed by the active nodes and the corresponding network links connecting these nodes. Note that it is necessary to remove the nodes that are not capable of carrying the flow, that is, when the flow is carried by the nodes, at least one dimension of the resource capacities of the nodes will be violated, leading to node congestion. We denote by $\mathcal{G}_c^m$ the residual network after removing the incapable nodes and the links attached to these nodes. We then carry out function $\mathtt{IsConn}$, a depth-first search procedure, to verify if the source and the destination of the current flow are connected in $\mathcal{G}_c^m$. If a candidate flow $\mathsf{d}_c$ that can be routed on $\mathcal{G}_c^m$ is found, we stop the search procedure; otherwise we pick up a flow demand uniformly at random (function $\mathtt{RandSelect}$) from the flow demand list. At this time, the residual capacity of the subnetwork formed by current active nodes is not sufficient for carrying any new flow, thus more nodes are needed to be activated so that routing demands for the newly selected flow can be satisfied. Once a candidate flow $\mathsf{d}_c$ has been determined, we remove the incapable nodes (those that satisfy $\vec{S}_v \lessdot \vec{R}_c$ which means that these exists at least one dimension $k$ such that $\vec{S}_v(k) > \vec{R}_c(k)$) according to the resources demand of the candidate flow and we denote by $\mathcal{G}_c$ the resulted network. Then, we apply a weight assignment process where we assign weights to the active nodes in $\mathcal{V}_a$ by invoking procedure $\mathtt{InvCount}$ (see below), and the weights for other nodes to be $(K(K-1)/2 + 1)$. In order to facilitate path selection, we carry out a node-link transformation procedure to assign weights for links based on the weights for nodes. The design of node weight assignment and node-link transformation will be detailed later in this section. At last, the candidate flow will be routed by involving a shortest-path-based algorithm such as Dijkstra algorithm on the weighted network $\mathcal{G}_c$ and will be removed from the demands list. The above process is repeated until the route for every flow has been assigned.  

\begin{algorithm}[t]
\caption{Multi-Resource Green (MRG) routing}\label{alg:mrg}
\begin{algorithmic}[1]
\State{$\mathcal{V}_a = \emptyset$;~~/*\textit{set of active nodes}*/}
\State{\textbf{for each} ($v \in \mathcal{V}$) $\vec{S}_v = \{0\}^K$;~~/*\textit{residual resources}*/}
\State{$\mathcal{E}_a \triangleq \{(v_1, v_2) \in \mathcal{E}~|~\forall v_1, v_2 \in \mathcal{V}_a\}$; $\mathcal{G}_a \triangleq \{\mathcal{V}_a, \mathcal{E}_a\}$;}
\State{\textbf{while} ($\mathcal{D}$ is not empty)}
	\State{~~$\mathsf{d}_c$ == $\mathrm{none}$;}
	\State{~~\textbf{for each} ($\mathsf{d}_m \in  \mathcal{D}$)~/*\textit{Search for a candidate flow}*/}
		\State{~~~~$\mathcal{G}_{c}^m = \mathcal{G}_a  \setminus \{v~|~\vec{S}_v \lessdot \vec{R}_m\}$;~/*\textit{remove incapable nodes}*/}
		\State{~~~~\textbf{if} ($\mathtt{IsConn}(\mathcal{G}_c^m, v_m^\mathrm{s}, v_m^\mathrm{t})$ == $\mathrm{true}$)}
			\State{~~~~~~$\mathsf{d}_c = \mathsf{d}_m$; $\mathcal{G}_c = \mathcal{G}_c^m$;}
			\State{~~~~~~\textbf{break};}
	\State{~~\textbf{if} ($\mathsf{d}_c$ == $\mathrm{none}$)~/*\textit{candidate flow not found}*/}
		\State{~~~~$\mathsf{d}_c = \mathtt{RandSelect}(\mathcal{D})$;}
	\State{~~$\mathcal{G}_c = \mathcal{G} \setminus \{v~|~\vec{S}_v \lessdot \vec{R}_c\}$;}
	\State{~~\textbf{for each} ($v \in \mathcal{V}$)~/*\textit{node weight assignment}*/}
	\State{~~~~\textbf{if} ($v\in \mathcal{V}_a$)~$w_v = \mathtt{InvCount}(\vec{S}_v, \vec{R}_m)$;}
	\State{~~~~\textbf{else}~$w_v = K(K-1)/2 + 1$;}
	\State{~~\textbf{for each} ($(v_1, v_2) \in \mathcal{E}_c$)~$w_e = (w_{v_1} + w_{v_2})/2$;}
	\State{~~$\mathcal{P}_c = \mathtt{SPath}(\mathcal{G}_c, \mathsf{d}_c)$;~/*\textit{shortest path routing}*/}
	\State{~~$\mathcal{V}_a = \mathcal{V}_a \cup \mathcal{P}_c$; $\mathcal{D} = \mathcal{D} \setminus \mathsf{d}_c$;}
	\State{~~\textbf{for each} ($v \in \mathcal{P}_c$)~$\vec{S}_v = \vec{S}_v - \vec{R}_m$;}
\end{algorithmic}
\end{algorithm}

{\bf Inversion-based node weight assignment.}~We now describe the function $\mathtt{InvCount}$ for assigning weights to network nodes. The second observation we mentioned at the beginning of this section suggests that once we have obtained a candidate flow to route on the subnetwork comprised of the active capable nodes, it is important to decide which nodes are preferable to carry the candidate flow. We provide a measure based on the distributions of the load vectors of both the node residual capacities and the flow demand. The general notion is that if the resource dimensions of a node are all kept balanced, then more flows will likely fit into the node. As a consequence, the number of nodes that need to be active will be reduced. To clarify, we first introduce the concept of \emph{inversion}.

\begin{definition}
Given two vectors $\vec{X} = \langle x_1,...,x_n \rangle$ and $\vec{Y} =  \langle y_1,...,y_n \rangle$, an \textbf{inversion} is defined as the condition $x_i > x_j$ and $y_i < y_j$, $1 \leq i,j \leq n$.
\end{definition}

\begin{property}
Given two vectors in $n$ dimensions, the total number of inversions is upper bounded by $n(n-1)/2$.
\end{property}

As we are focusing on the distributions of the node residual capacities and the flow resources demands, it is straightforward that an inversion can lead to much heavier resource dimensions imbalance on a node as the scarce resource is demanded more and the abundant resource is demanded less. Therefore, in order to keep all the dimensions of resources balanced, the number of inversions has to be minimized. Based on this principle, the inversion-based node weight assignment procedure assigns weights for nodes that are already active according to the number of inversions shared by the node residual capacity vector and the flow demand vector. The weights of the inactive nodes are set to be one unit larger than the maximum number of inversions that can be shared by any residual capacity vector and flow demand vector. As a result, if possible, the nodes that are active and with less numbers of inversions will be preferably chosen to carry the candidate flow and the inactive nodes have the lowest priority to be used.

{\bf Path selection.}~The guideline for selecting the route for the candidate flow is to choose a path that connects the source and the destination of the candidate flow while minimizing the total weights of nodes that are on the path. This is actually equivalent to solving a node-weighted single-source shortest path routing problem. We notice that this problem can be transformed into a traditional link-weighted single-source shortest path routing problem by setting the weight of each link to be the half of the sum of the weights of the endpoints of this link. Denote by $\mathbb{R}_1$ and $\mathbb{R}_2$ the node-weighted and the transformed link-weighted shortest path routing problems respectively. We have the following property.
\begin{property}
Solving $\mathbb{R}_1$ is equivalent to soving $\mathbb{R}_2$.
\end{property}

This is because as long as a flow is routed through an intermediate node (nodes except the source and the destination) on a path, the weight on this node will be shared by two links (i.e., the ingress and egress links). Thus, if we let all the links that are attached to an intermediate node to share half of the node weight, it is always true that the total weights on the links on this path will be equal to the total weights on the internal nodes on this path. As a result, solving the corresponding link-weighted single-source shortest path routing problem will also give solutions to the path selection for the candidate flow. It is well-known that the link-weighted single-source shortest path routing problem can be solved efficiently by using the Dijkstra algorithm.

\subsection{Time Complexity}

We now analyze the time complexity of the proposed MRG algorithm. The algorithm runs iteratively and in each iteration exactly one flow will be chosen as the candidate and will be routed. As a result, the maximum number of iterations will be upper bounded by the number of flow demands $M$. In each iteration, the algorithm first searches in the flow demand list to find out a candidate flow and the most time consuming part in this candidate flow searching procedure is depth-first search which can be accomplished in $\mathcal{O}(|E|)$ time where $|E| \leq N^2$ is the total number of edges in the network ($N$ is the total number of nodes). The total searching time in one iteration then will be in $\mathcal{O}(M\cdot |E|)$. Once the candidate flow is found, the time complexity will be dominated by the shortest path routing algorithm which can be done in $\mathcal{O}(|E| + N \cdot \log N)$ time as the weight assignment procedure can be finished in time $\mathcal{O}(N \cdot K \cdot \log K)$. Combining all these, we have that the MRG algorithm can be finished in $\mathcal{O}(|E|M^2 )$ time.




\section{Topology-aware Heuristic}
\label{sec:heuristic}

\begin{algorithm}[t]
\caption{Hierarchical Green Routing (HGR)}\label{alg:hgr}
\begin{algorithmic}[1]
\State \textbf{function} $\mathtt{VBP}$($\mathcal{D}$)~/*\textit{vector bin packing algorithm}*/
	\State ~~$\mathrm{idx} = 1$; $\mathsf{d}_c = \mathrm{none}$; $S_{\mathrm{idx}} = \{0\}^K$;
	\State ~~$\alpha_k = \sum_{\mathsf{d}_m \in \mathcal{D}} R_m(k) / \sum_{\mathsf{d}_m \in \mathcal{D}} \sum_{k=1}^K R_m(k)$;
	\State ~~\textbf{while} ($\mathcal{D}$ is not empty)
	\State ~~~~$\mathcal{D}_c = \mathcal{D} \setminus \{\mathsf{d}_m \in \mathcal{D}~|~S_{\mathrm{idx}} \lessdot R_m \}$;
	\State ~~~~$\mathsf{d}_c = \argmin_{\mathsf{d}_m \in \mathcal{D}_c} \sum_{k=1}^K \alpha_k (S_{\mathrm{idx}}(k) - R_m(k))^2$;
	\State ~~~~\textbf{if} ($\mathsf{d}_c$ == $\mathrm{none}$)~$\mathrm{idx} $++;~/*\textit{open a new bin}*/
	\State ~~~~\textbf{else}~$\mathcal{D} = \mathcal{D} \setminus \{\mathsf{d}_c\}$;~/*\textit{pack the current item}*/
	\State ~~\textbf{return} $\mathrm{idx}$
\State \textbf{for} ($0 \leq i \leq z-1$)~/*\textit{\# of aggr. nodes in each pod}*/
\State ~~$N_i^{\mathrm{agg}} = \mathtt{VBP}(\{\mathrm{d}_m~|~v_m^{\mathrm{s}}~\mathrm{or}~v_m^{\mathrm{t}}~\mathrm{in~pod}~i\})$
\State \textbf{for} ($0 \leq j \leq z/2 - 1$)~/*\textit{\# of core nodes}*/
\State ~~$N_j^{\mathrm{core}} = \mathtt{VBP}(\{\mathrm{d}_m~|~(v_m^{\mathrm{s}}~\mathrm{or}~v_m^{\mathrm{t}}~\mathrm{mod}~(z^2/4))/ 2 = j\})$
\end{algorithmic}
\end{algorithm}

The proposed MRG algorithm can leverage the coordination of the flow demands in multiple dimensions and minimize the number of active network nodes efficiently. However, MRG is generally conducted without taking into account the topology features of the network. We notice that topologies of the networks commonly used in data center networks such as fat-tree or VL2 have very high level of symmetry and they are usually well structured in layers. Therefore, we argue that the routing algorithm can be further improved by taking advantage of the topology characteristics. In this section, we provide a new topology-aware heuristic for the most common tree-like data center network topologies. 


The key observation we have from tree-like topologies is that the number of active nodes can be determined layer by layer. We take a typical fat-tree topology (as shown in Fig.~\ref{fig:fat-tree}) as an example. The number of edge nodes cannot be optimized since edge nodes are also responsible for inter-host communication in the same rack. In each pod, the number of aggregation nodes can be determined according to the flow demands that flow out of and into the pod. This is actually to solve a vector bin packing problem as we have introduced previously. The core layer is a bit different from the aggregation layer; for a $z$-$\mathrm{ary}$ fat-tree, all the cores nodes that share congruence with respect to $(z/2)$ will be responsible for carrying the flow demands from the aggregation nodes in the same positions in every pod. Thus for these core nodes, solving a vector bin packing can give the right number of nodes that need to stay active. Inspired by this observation, we propose HGR, a hierarchical energy-efficient routing algorithm based on solving a set of vector bin packing problems. The pseudocode of HGR is shown in Algorithm~\ref{alg:hgr}.

{\bf Vector bin packing.}~The function $\mathtt{VBP}$ we adopted for solving the vector bin packing problem is a norm-based greedy algorithm \cite{vbp}. The algorithm is bin-centric which means that it focuses on one bin $\mathrm{idx}$ and always places the most suitable remaining item that fits in the bin. To find out the most suitable item, the algorithm looks at the difference between the demand vector $R_m$ and the residual capacity vector $S_{\mathrm{idx}}$ under a certain norm. We choose the $\ell_2$-norm and from all unassigned items, we choose the item that minimizes $\sum_{k=1}^K \alpha_k (S_{\mathrm{idx}}(k) - R_m(k))^2$ where $\alpha_k$ represents the importance of dimension $k$ among all dimensions and is given by
$$\alpha_k = \frac{\sum_{\mathsf{d}_m \in \mathcal{D}} R_m(k)}{\sum_{\mathsf{d}_m \in \mathcal{D}} \sum_{k=1}^K R_m(k)}.$$
If no item can be found to fit into the current bin $\mathrm{idx}$, we open a new bin and repeat the above procedure.

\begin{figure}[!t]
\centering
\includegraphics[scale=0.32]{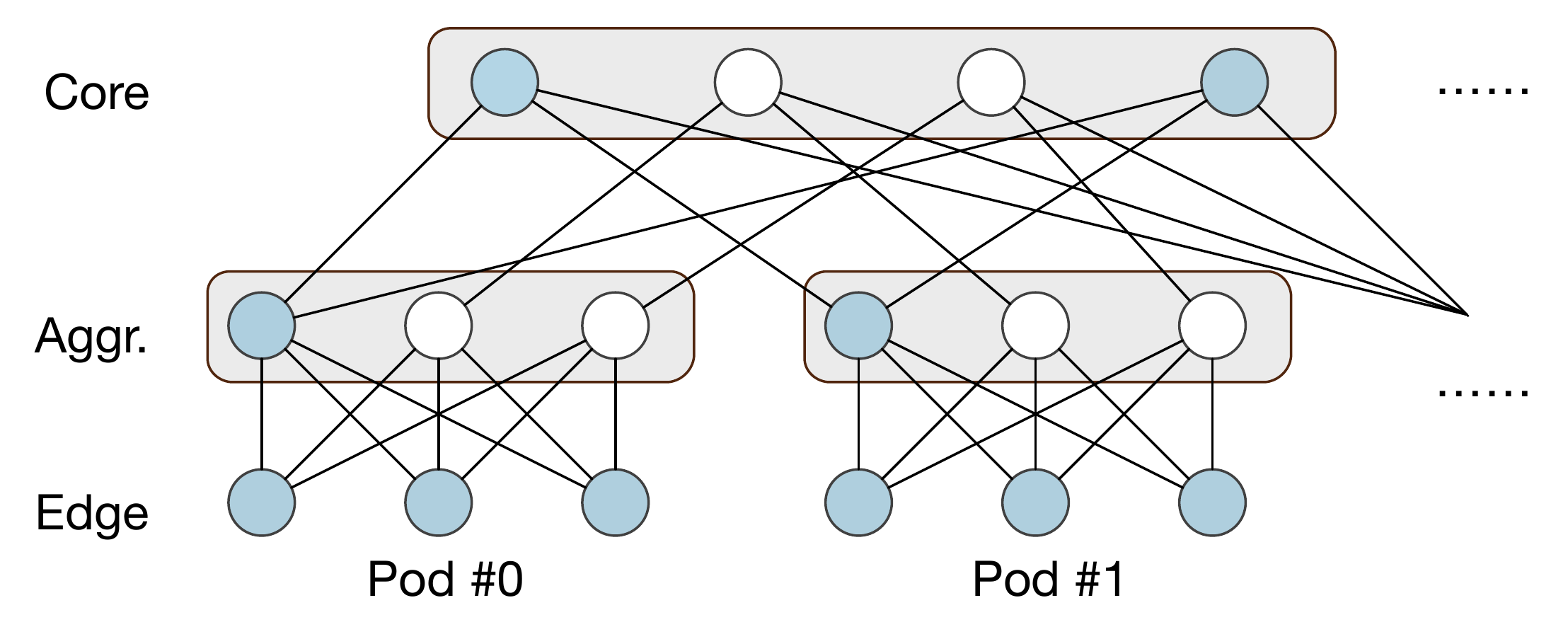}
\caption{\label{fig:fat-tree} An example to show the topology-aware heuristic in a fat-tree topology, where the set of active nodes are determined layer by layer.}
\vspace{-0.5cm}
\end{figure}

{\bf Time complexity.} The HGR algorithm replies on solving several instances of the vector bin packing problem. In the worst case, the sizes of the vector bin packing instances can be as large as $\mathcal{O}(M)$ and thus it will take $\mathcal{O}(M^2)$ time to be solved by $\mathtt{VBP}$ algorithm. As a result, the total time complexity of HGR can be given by $\mathcal{O}(M^2)$. Compared to the MRG algorithm, HGR can provide a speedup of $\Omega(|E|)$. We will validate this speedup by simulations.

\section{Evaluation}
\label{sec:simu}

We carried out extensive simulations to evaluate the performance of the proposed algorithms. In this section, we provide a detailed summary of our simulation findings.

\subsection{Simulation Settings}

We deploy our algorithms on a laptop with a Core i5 2.6GHz CPU with two physical cores and 8GB DRAM. All of the algorithms are implemented in Python.

We choose fat-trees of different sizes as the data center network topologies. This is because fat-tree is a typical topology used in DCNs, and can provide equal-length parallel paths between any pair of end hosts, which is very beneficial for software packet processing paradigm to embed processing functions into the routing paths regardless of the topology details. The flow demands we used in our simulations are generated randomly: the endpoints of each flow are chosen uniformly at random from the set of end hosts. The requirement of each resource dimension of each flow is generated following a normal distribution (in the positive side) where the mean and the variation are all set to be $0.02$ to provide large resource demand diversity. The node capacity of each resource dimension is assumed to be normalized to $1$. 

We carry out two groups of simulations for validating MRG and HGR respectively: $\mathit{i})$ For evaluating the performance of algorithm MRG, we compare it with three other algorithms of interest: Single-Resource Shortest Path (SRSP), Single-Resource Green (SRG), Multi-Resource Shortest Path (MRSP). The efficiency of energy saving of the four algorithms are examined on two fat-tree topologies in different scales under different numbers of flow demands. We also explore the impact of the number of resource dimensions under certain scenarios. $\mathit{ii})$ The performance of algorithm HGR is compared with that of algorithm MRG. We first study the impact of the number of resource dimensions. Then, under certain scenarios, we examine the efficiency of energy saving and the running time of both MRG and HGR under different numbers of flow demands. All the results are averaged among $20$ independent tests and all the figures show with the average and the standard deviation.

\subsection{Performance of Algorithm MRG}

\begin{figure}
	\centering
	\subfigure[Energy savings under different numbers of flows]{
 	\label{fig:8_5:es} 
	\includegraphics[scale=0.23]{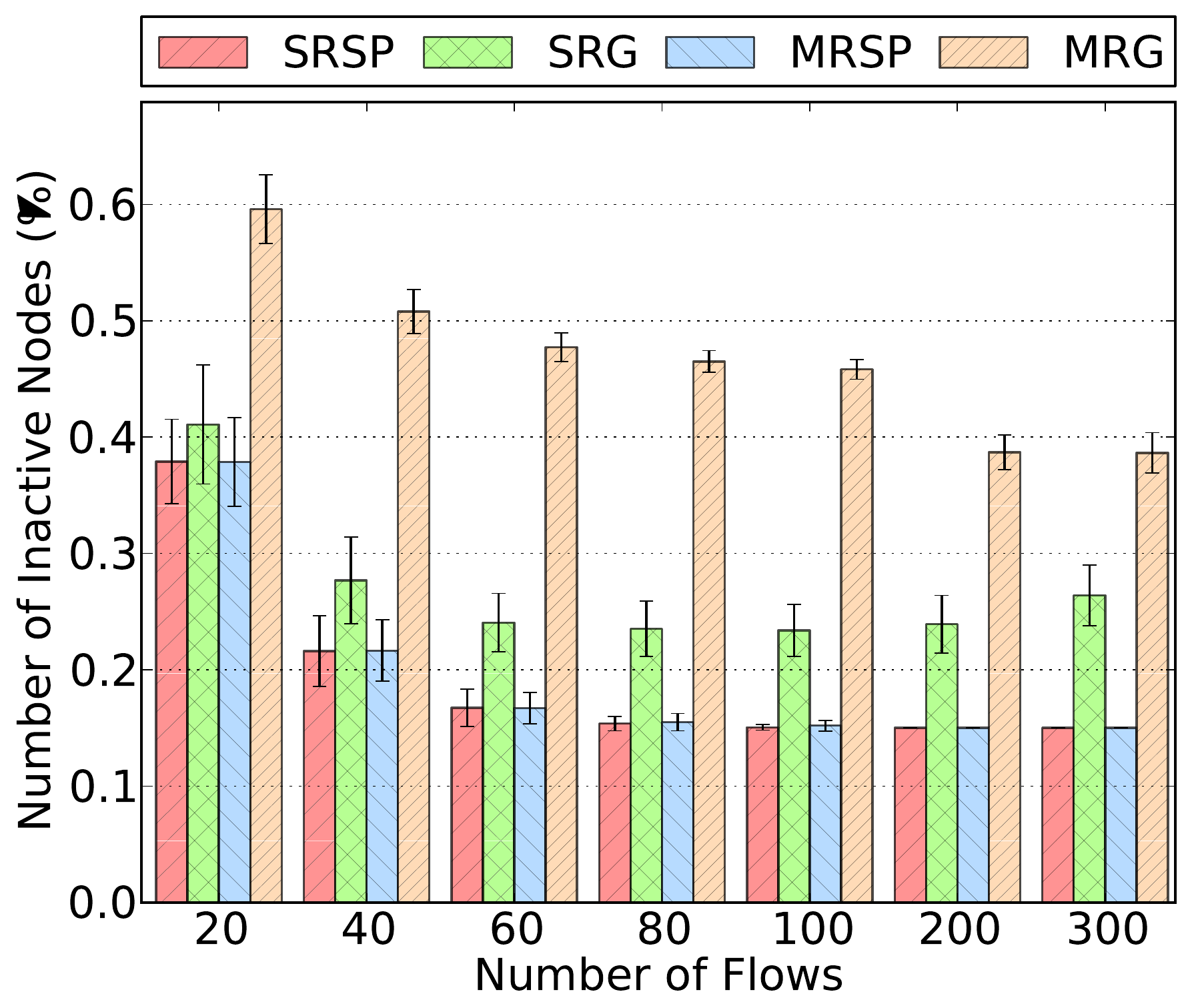}}
	\hspace{-0.1in}
	\subfigure[Number of incomplete flows]{
	\label{fig:8_5:inc} 
	\includegraphics[scale=0.23]{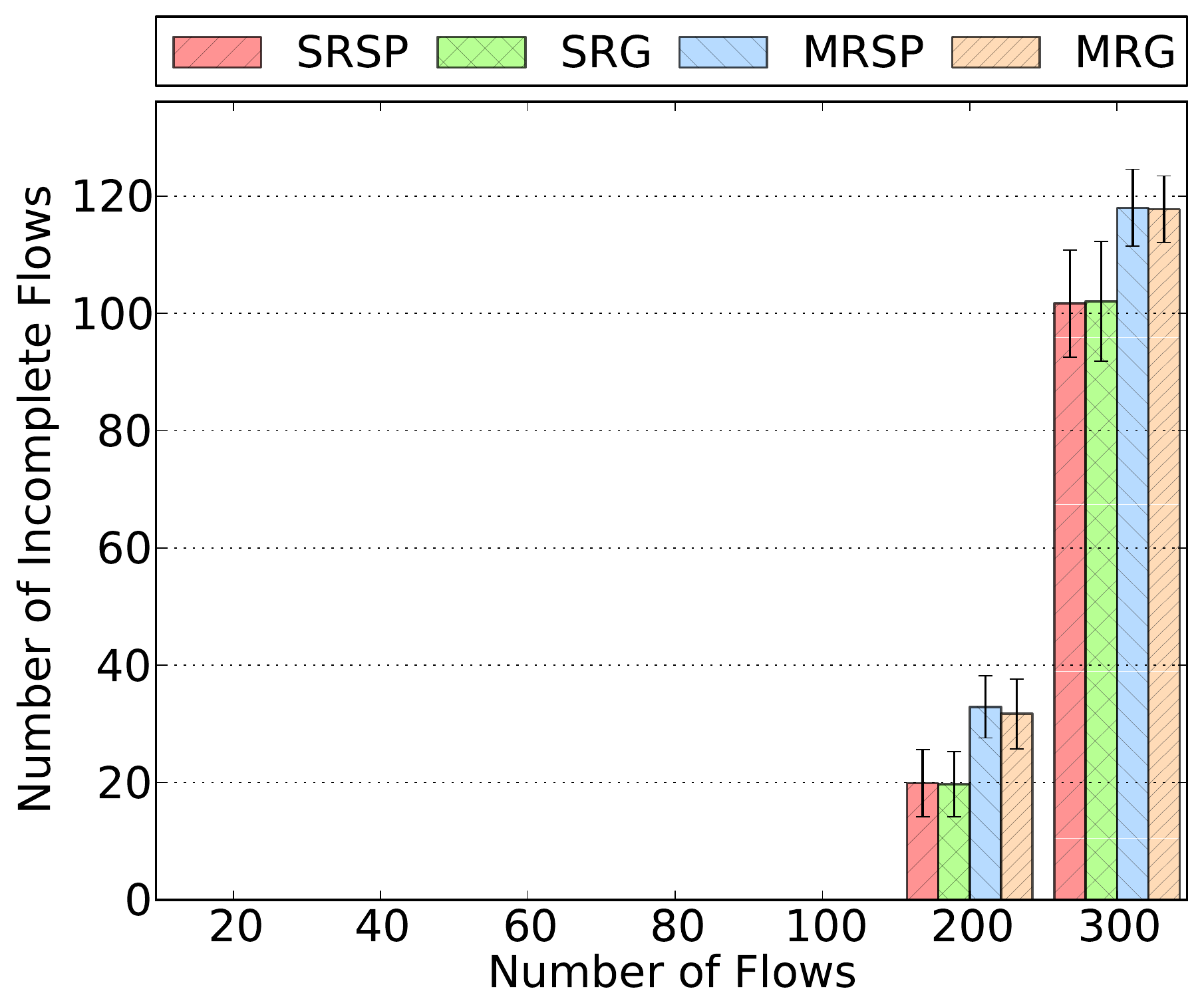}}
	\hspace{-0.1in}
	\subfigure[Number of congested nodes]{
 	\label{fig:8_5:vio} 
	\includegraphics[scale=0.23]{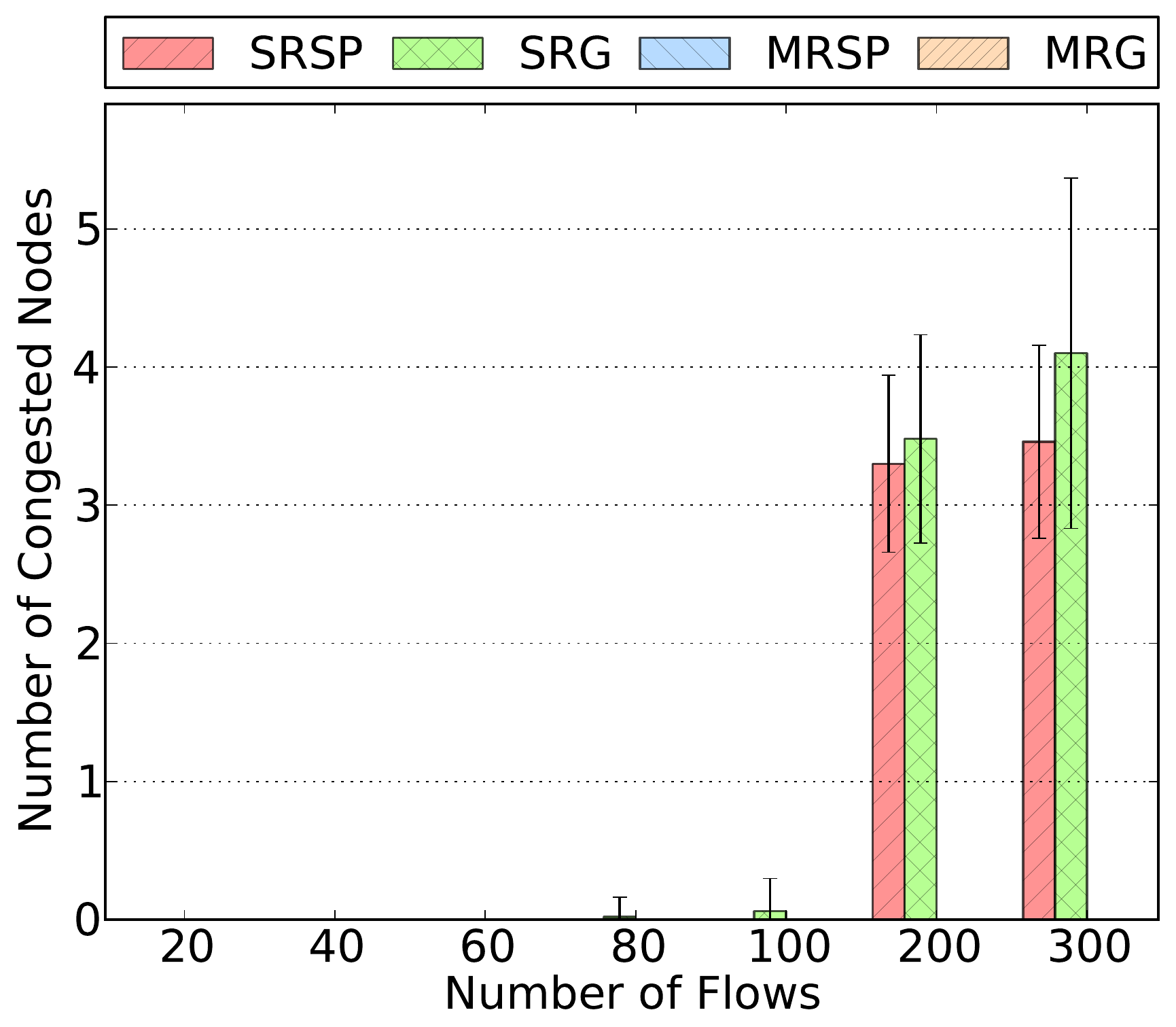}}
	\hspace{-0.1in}
	\subfigure[Energy savings under different numbers of resource dimensions]{
 	\label{fig:8:dim} 
	\includegraphics[scale=0.23]{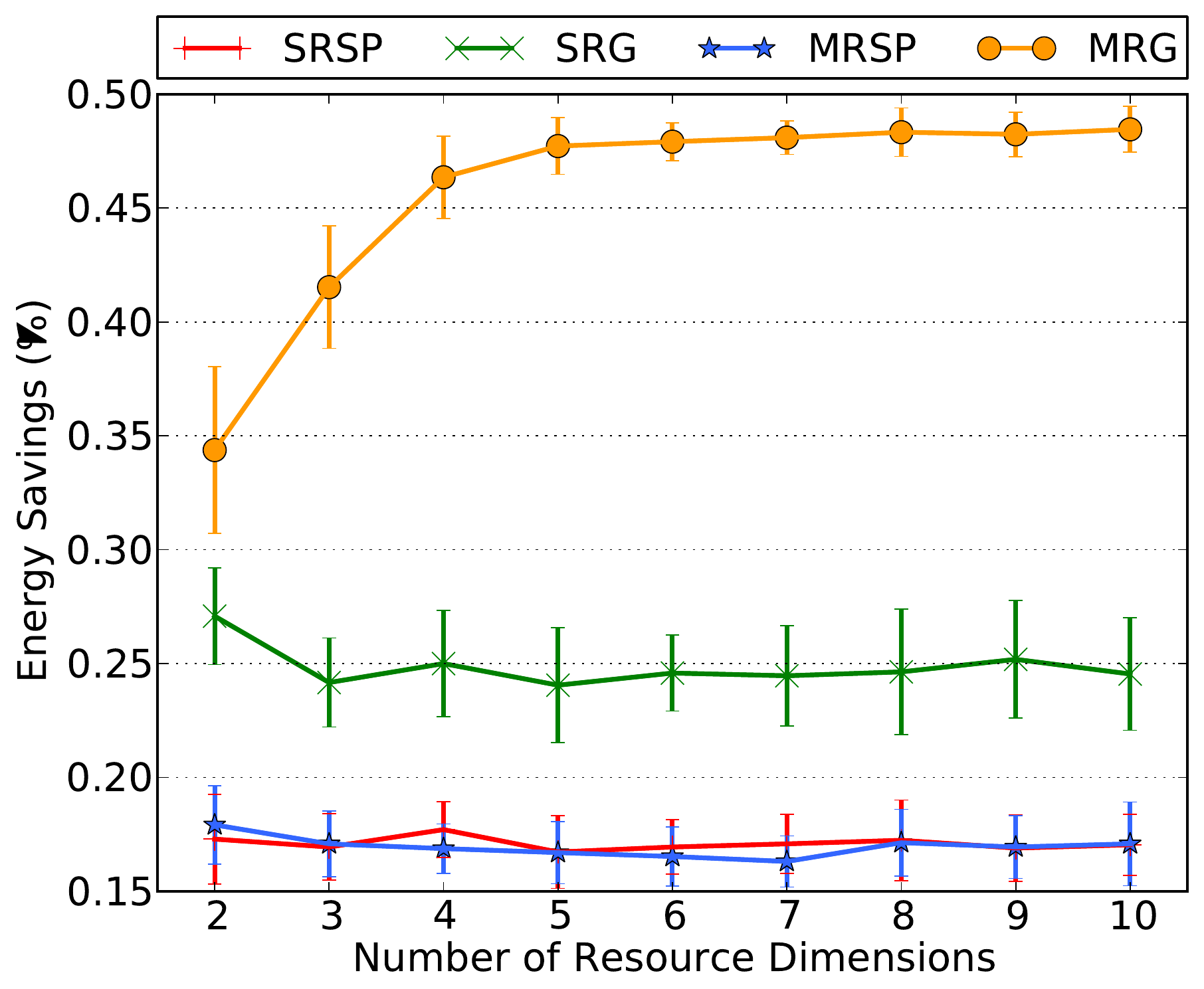}}
	\caption{\label{fig:mrg-8}Performance comparison for MRG under the scenarios where the network topology is given by an 8-$\mathrm{ary}$ fat-tree with 208 nodes (128 end-hosts and 80 packet processors).}
\vspace{-0.4cm}

\end{figure} 

{\bf Energy savings.} The simulation results for evaluating the energy saving performance of MRG are depicted in Fig.~\ref{fig:mrg-8}(a, b, c). The energy saving ratio is represented by the number of inactive nodes divided by the total number of nodes. It can be seen from Fig.~\ref{fig:8_5:es} that MRG outperforms the other three algorithms with respect to energy savings under all scenarios. SRSP and MRSP converge to very low energy saving ratios very quickly while SRG and MRG can exploit more energy saving potentials by carefully steering traffic flows. We also compare the performance of all the algorithms under extremely heavy load scenarios. When the number of flow demands exceeds the capability of the network (and congestion happens at some critical nodes), MRSP and MRG will block more flows than SRSP and SRG as can be seen in Fig.~\ref{fig:8_5:inc}. This is reasonable because MRSP and MRG take into account more resource dimensions and it is likely that node capacities are violated more easily than with single-resource solutions. However, when considering only one resource dimension, some nodes will be congested due to the neglect of other resource dimensions, although more flow demands are likely to be assigned. The numbers of nodes that are congested under different numbers of flows are shown in Fig.~\ref{fig:8_5:vio}.

{\bf Impact of the number of resource dimensions.}~Fig.~\ref{fig:8:dim} depicts the simulation results for examining scalability of MRG with respect to the number of resource dimensions. It can be obviously noticed that the energy saving performance of MRG has a very significant improvement with the increase of the number of resource dimensions and converges to a high level. This is because with more resource dimensions, the proposed inversion-based node weight assignment can distinguish nodes from one another more accurately and thus the path chosen for each flow will be more effective in terms of energy saving.

\subsection{Performance of Algorithm HGR}

We first compare the scalability of MRG and HGR with respect to the number of resource dimensions. The simulation results are shown in Fig.~\ref{fig2:8:dim}. It can be observed that HGR outperforms MRG when the number of resource dimensions is very small. However, with the increase of the number of resource dimensions, the energy saving performance of HGR drops dramatically with a constant rate, while MRG performs better and better and converges finally as we have discussed before. This is mainly because HGR is largely based on the vector bin packing heuristic which performs well when the number of dimensions is small due to the greedy manner of item assignment, but it has very poor scalability with respect to the number of dimensions.

We then choose a fair number of resource dimensions ($K = 3$) and compare both the energy saving ratio and the running time of MRG and HGR. The energy saving results are depicted in Fig.~\ref{fig2:8:es}. We observe that when the number of flow demands is not very large, MRG and HGR are comparable in terms of energy savings, but HGR suffers from some performance degradation when the number of flows is very large. However, HGR compensates this slight loss of energy efficiency by a very significant reduction on the running time. As can be seen from Table~\ref{tb:time}, for a fat-tree with $80$ packet processing nodes (i.e., $|E| = 192$), the running time of HGR is around $0.5$ percent of that of MRG, which confirms the lower bound on the speedup $\Omega(|E|)$.

\begin{figure}
	\centering
	\subfigure[Energy savings under different numbers of resource dimension]{
 	\label{fig2:8:dim} 
	\includegraphics[scale=0.23]{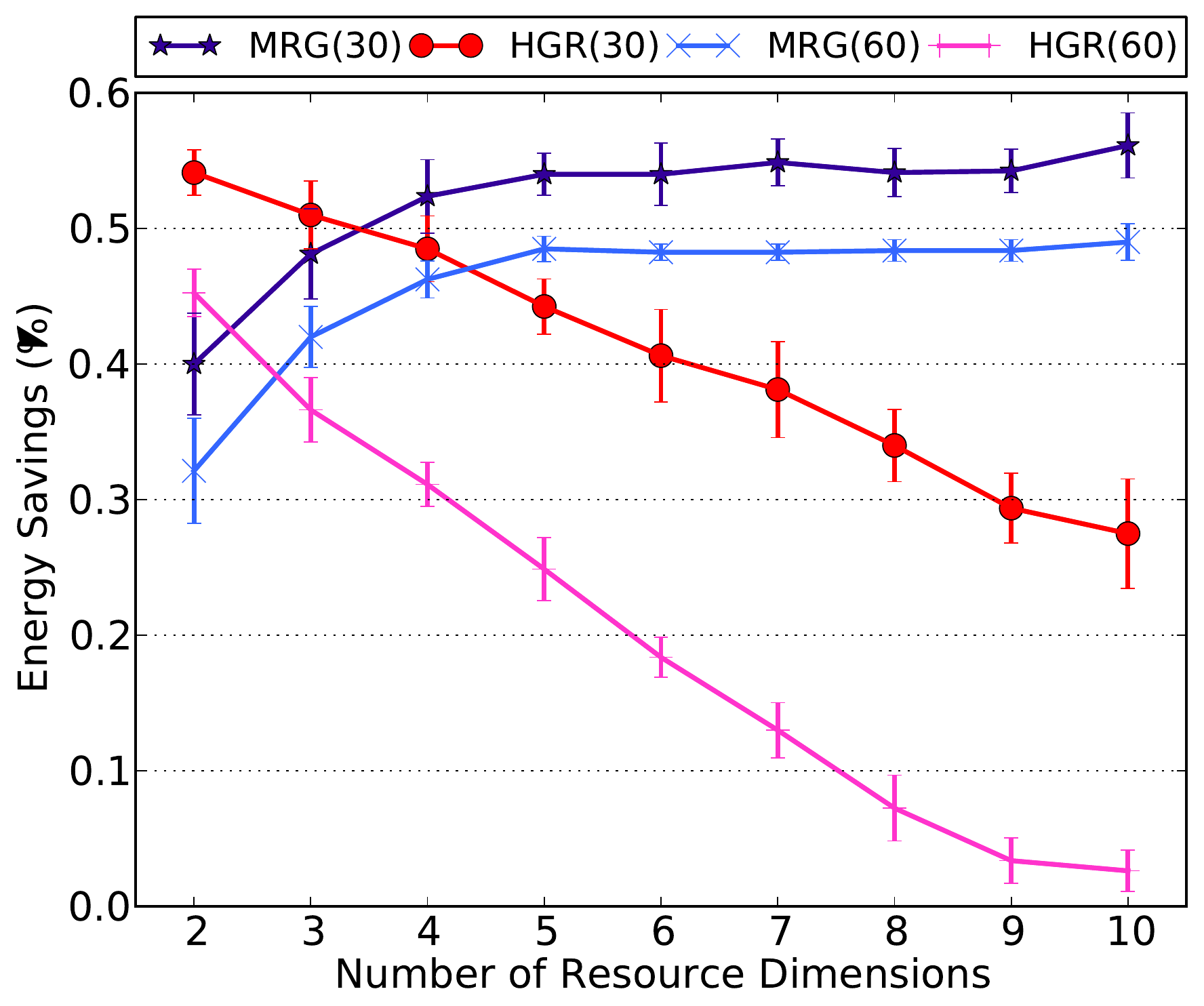}}
	\hspace{-0.05in}
	\subfigure[Energy savings under different numbers of flows]{
	\label{fig2:8:es} 
	\includegraphics[scale=0.23]{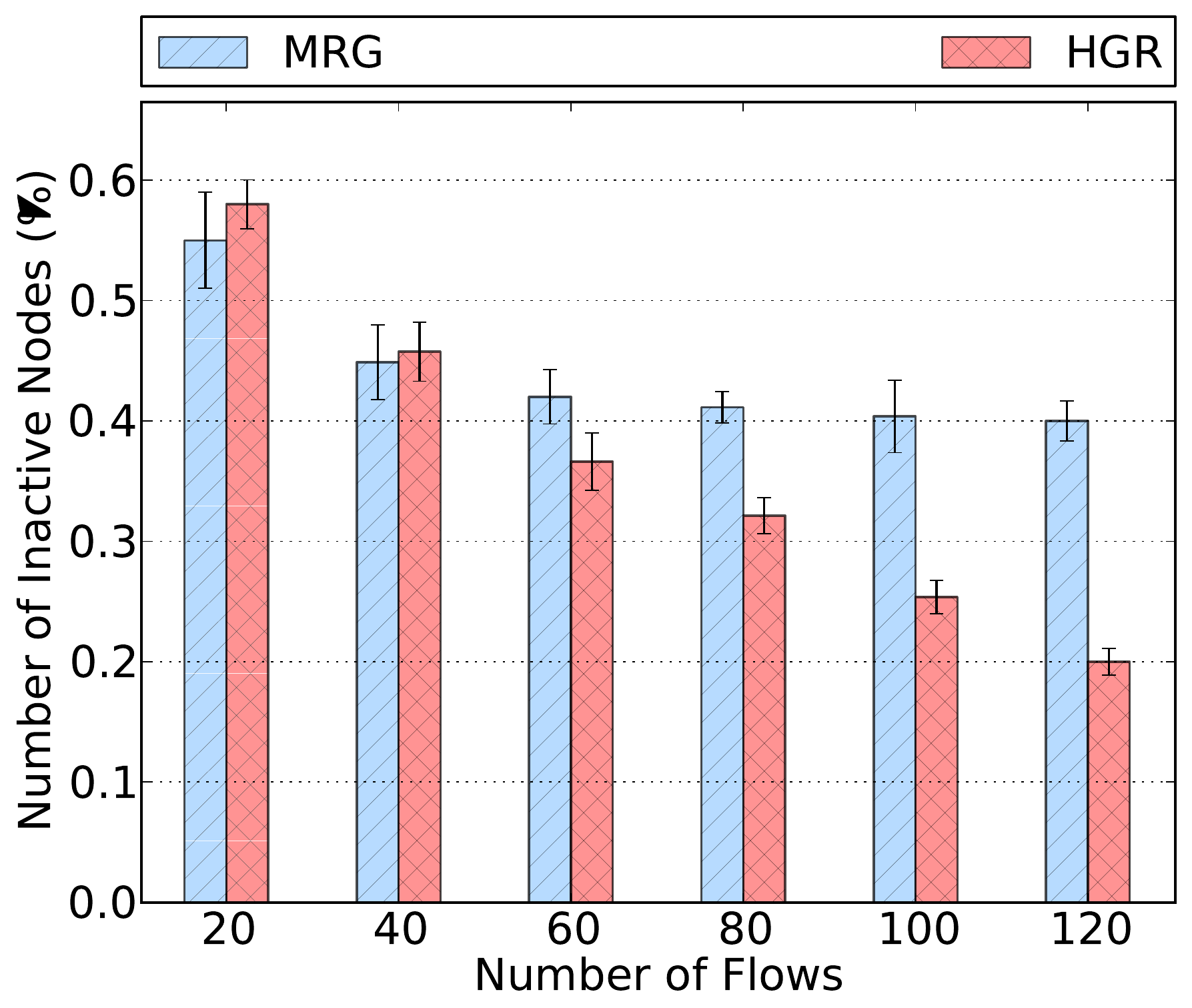}}
	\caption{\label{fig:hgr}Performance comparison for HGR under the scenarios where the network topology is given by an 8-$\mathrm{ary}$ fat-tree with 208 nodes (128 end-hosts and 80 packet processors).}
\vspace{-0.5cm}
\end{figure} 

\begin{table}
\caption{\label{tb:time} Running time Statistics of the Algorithms (Unit: secs)}
\centering
\begin{tabular}{ c  c  c  c  c  c  c  }
  \toprule
  $\mathrm{\#~of~flows}$ & 20 & 40 & 60 & 80 & 100 & 120 \\
  \midrule
  $\mathrm{alg.~MRG}$  & 5.37  & 16.63  & 37.00 & 58.26 & 92.63 & 101.89  \\
  \midrule
  $\mathrm{alg.~HGR}$ & 0.026 & 0.078 & 0.192 & 0.400 & 0.647 & 0.681 \\
  \bottomrule
\end{tabular}
\vspace{-0.5cm}
\end{table}

\section{Discussion}
\label{sec:disc}

We discuss in this section some practical issues that are related to the application of the proposed technique.

\subsection{Model Extension}
{\bf Dynamic flow joining and leaving.}~The problem we have discussed is for scenarios where a static set of flow demands is given a priori and the proposed MRG algorithm is dedicated to solving this problem. However, the reality differs from the static case by having flows joining and leaving the network dynamically. We observe that although we did not take into account the dynamic property of the set of flow demands, the MRG algorithm can be extended to the online case due to its progressive fashion. When a new flow arrives in the network, we first check whether the subnetwork formed by only the active nodes is capable of carrying this flow. If it is true, we carry out the node-weight assignment and path selection procedures to find out a routing path and route the flow using this path. Otherwise, we include also the inactive nodes with weights assigned and find out a path in the resulting network. When a flow completes its transmission, we focus on two types of flows: for the existing flows that have very short life-times, we leave them as they are in the network as they will be completed in a short time; for the long-lived flows in the network, we buffer the newly arrived flows until the existing short-lived flows are gone and we carry out the MRG algorithm to reroute those flows that have long life-times in order to achieve energy efficiency. After this rerouting, the routes for the buffered new flows will be assigned as well. The node-weight assignment and path selection procedures have low complexity. Thus they can be applied to online scenarios conveniently. Nevertheless, the centralized environment also enables parallel acceleration to ensure realtime optimization.

{\bf Heterogeneity.}~For the sake of tractability, we assumed in our model that the resources required for processing all the packets from a given flow on every node of the flow's routing path are very related to the size of the packets, and thus homogeneity can be assumed. However, this may not be true when a flow requires different processing functions on different nodes. Our model can be extended to the heterogeneous case by treating each node on the routing path independently, e.g., in the MRG algorithm, the weights on nodes can be assigned in a hop by hop manner; in the HGR algorithm, each vector bin packing instance is solved by having different resource requirements from the flow demands. We leave more elaborate solutions for future work.



\subsection{Practical Application Scenarios}

{\bf Named data networking.}~To generalize the role of thin waist of the IP architecture, Named Data Networking (NDN) was proposed, where packets can name objects other than communication endpoints. NDN routes and forwards packets based on names, which requires high-performance processing (e.g., prefix matching) capability at network nodes. Moreover, it is also useful for a network node to cache the received data packets in its content store and use them to satisfy future requests. These properties make NDN a good application scenario for the NaaS model. As a result, the proposed GreenNaaS solution will have the potential to be used for achieving energy efficiency in NDN.

{\bf Server-centric data center network architecture.}~In traditional switch-centric networks, data packets are transmitted through only proprietary network devices such as switches or routes. With the possibility of having multi-port network cards on servers, several server-centric network architectures such as BCube \cite{Guo_Lu-2009}, SWCube and SWKautz \cite{Li_Wu-2014} have been recently proposed for data center networks. In these server-centric network architectures, servers are also involved in packet forwarding. By applying more application-specific processing functions on packets, these architectures are very easy to be extended to adopt the NaaS model and thus GreenNaas can be used to save energy in those server-centric data center networks where energy issue is more prominent than in traditional switch-centric data center networks.

{\bf Middlebox orchestration.}
Middleboxes are special network devices that are responsible for packet processing to provide functionality such as NATs, firewalls or WAN optimizers. Traditionally, these devices are proprietary and have been implemented on closed hardware platforms which are usually hard to be extended. To change this situation, recently many proposals have been provided to make middlebox functionalities software-centric (e.g., \cite{Sekar_Egi-2012, Gember_Viswanathan-2014}). NaaS incorporates these by leveraging application-specific in-network packet processing. Although there are still many critical problems, such as service chain design, needing more research efforts, GreenNaaS provides a clear insight on how to achieve energy efficiency in NaaS systems which is definitely an important issue in the near future.
\section{Conclusion}
\label{sec:conc}
We study the energy-efficiency multi-resource routing  problem which arises from the recently proposed cloud networking model NaaS. This optimization problem differs from the traditional energy-efficient routing problem by having node capacities and flow demands represented by vectors in multiple dimensions. We provide a simple iterative routing scheme which selects flows iteratively to exhaust the residual capacities in active nodes and assign routes to flows based on the distributions of node residual capacities and flow demands. To leverage the structural property of data center network topologies, we also provide a topology-aware heuristic designated to fat-trees, which can provide comparably good energy efficiency while significantly reducing the computation time. 






%
%
%

\bibliographystyle{IEEEtran}
\bibliography{IEEEabrv,ref}

\end{document}